\documentclass[a4paper,UKenglish]{lipics-v2018}

\usepackage{microtype}
\title{Largest Weight Common Subtree Embeddings with Distance Penalties}
\titlerunning{Largest Weight Common Subtree Embeddings}

\author{Andre Droschinsky}{Department of Computer Science, TU Dortmund University\\{Otto-Hahn-Str.\,14, 44221 Dortmund, Germany}}{andre.droschinsky@tu-dortmund.de}{0000-0002-7983-3739}{}%
\author{Nils M.\@ Kriege}{Department of Computer Science, TU Dortmund University\\{Otto-Hahn-Str.\,14, 44221 Dortmund, Germany}}{nils.kriege@tu-dortmund.de}{0000-0003-2645-947X}{}%
\author{Petra Mutzel}{Department of Computer Science, TU Dortmund University\\{Otto-Hahn-Str.\,14, 44221 Dortmund, Germany}}{petra.mutzel@tu-dortmund.de}{0000-0001-7621-971X}{}%

\authorrunning{A.~Droschinsky, N.~Kriege and P.~Mutzel}%
\Copyright{Andre Droschinsky, Nils M.\@ Kriege and Petra Mutzel}%

\subjclass{\vspace{-.1em}}%


\funding{This work was supported by the German Research Foundation (DFG), priority programme ``Algorithms for Big Data'' (SPP 1736).}%

\nolinenumbers
\hideLIPIcs

\usepackage{tikz}
\tikzstyle{edge}=[draw=black, thick,-]
\tikzstyle{vertex}=[circle, draw=black,inner sep=3.5pt,semithick]
\tikzstyle{vertexWhite}=[vertex, fill=white]
\tikzstyle{vertexLarge}=[vertex, inner sep=4.5pt]
\tikzstyle{vertexWhiteLarge}=[vertexLarge, fill=white]
\tikzstyle{edgeIso}=[draw=black, -, dotted, thick, bend angle=20, bend right]
\tikzstyle{edgeIsoGreen}=[draw=black!20!green, -,dashed, thick, bend angle=20, bend left]

\usepackage{xspace}
\def\MWM{MWM\xspace}
\def\ES{CSE\xspace}
\def\LES{LaCSE\xspace}
\def\WES{LaWeCSE\xspace}
\def\WESu{LaWeCSE$_\text{u}$\xspace}

\newcommand{\trtr}[0]{\ensuremath{\curlywedge}\xspace}
\newcommand{\tsk}[0]{\ensuremath{\diamond}\xspace}
\newcommand{\tin}[0]{\ensuremath{t\in\{\trtr,\tsk\}}\xspace}

\newcommand{\wIso}[0]{\ensuremath{\omega}\xspace}
\newcommand{\wPath}[0]{\ensuremath{\omega_{\text{p}}}\xspace}
\newcommand{\WIso}[0]{\ensuremath{\mathcal{W}}\xspace}
\newcommand{\wMatch}[0]{\ensuremath{w}\xspace}
\newcommand{\WMatch}[0]{\ensuremath{W}\xspace}
\newcommand{\MWMk}[1]{\text{MWM$_{#1}$}}
\newcommand{\TabL}[0]{\ensuremath{\mathcal L}\xspace}

\begin{document}

\maketitle

\begin{abstract}
The largest common embeddable subtree problem asks for the largest possible tree embeddable into two input trees
and generalizes the classical maximum common subtree problem.
Several variants of the problem in labeled and unlabeled rooted trees have been studied, e.g., for the comparison of evolutionary trees.
We consider a generalization, where the sought embedding is maximal with regard to a weight function on pairs of labels. We support rooted and unrooted trees with vertex and edge labels as well as distance penalties for skipping vertices.
This variant is important for many applications such as the comparison of chemical structures and evolutionary trees.
Our algorithm computes the solution from a series of bipartite matching instances, which are solved efficiently by exploiting their structural relation and imbalance.
Our analysis shows that our approach improves or matches the running time of the formally best algorithms for several problem variants. 
Specifically, we obtain a running time of $\mathcal O(|T|\,|T'|\Delta)$ for two rooted or unrooted trees $T$ and $T'$,
where $\Delta=\min\{\Delta(T),\Delta(T')\}$ with $\Delta(X)$ the maximum degree of $X$.
If the weights are integral and at most $C$, we obtain a running time of $\mathcal O(|T|\,|T'|\sqrt\Delta\log (C\min\{|T|,|T'|\}))$ for rooted trees.
\end{abstract}

\section{Introduction}
The maximum common subgraph problem asks for a graph with a maximum number of 
vertices that is isomorphic to induced subgraphs of two input graphs.
This problem arises in many domains, where it is important to find 
the common parts of objects which can be represented as graphs.
An example of this are chemical structures, which can be interpreted directly as 
labeled graphs. Therefore, the problem has been studied extensively in
cheminformatics~\cite{EhRa2011,Raymond2002,Schietgat2013}.
Although elaborated backtracking algorithms have been developed~\cite{Raymond2002,McCreesh2017},
solving large instances in practice is a great challenge.
The maximum common subgraph problem is \textsf{NP}-hard and remains so even when the input graphs are restricted to trees~\cite{Garey1979}. However, in trees it becomes polynomial-time solvable when the common subgraph is required to be connected, i.e., it must be a tree itself. This problem is then referred to as \emph{maximum common subtree problem} and the first algorithm solving it in polynomial-time is attributed to J.~Edmonds~\cite{Matula1978}.
Also requiring that the common subgraph must be connected (or even partially biconnected) several extensions to tree-like graphs have been proposed, primarily for applications in cheminformatics~\cite{Yamaguchi2004,Schietgat2013,DrKrMu2017}. Some of these approaches are not suitable for practical applications due to high constants hidden in the polynomial running time. Other algorithms are efficient in practice, but restrict the search space to specific common subgraphs.
Instead of developing maximum common subgraph algorithms for more general graph classes, which has proven difficult, a different approach is to represent molecules simplified as trees~\cite{Rarey1998}.
Then, vertices typically represent groups of atoms and their comparison requires to score the similarity of two vertices by a weight function. This, however, is often not supported by algorithms for tree comparison. Moreover, it maybe desirable to map a path in one tree to a single edge in the other tree, skipping the inner vertices. Formally, this is achieved by \emph{graph homeomorphism} instead of isomorphism.

Various variants for comparing trees have been proposed and investigated~\cite{Valiente2002}.
Most of them assume rooted trees, which may be ordered or unordered. Algorithms
tailored to the comparison of evolutionary trees typically assume only the 
leaves to be labeled, while others support labels on all vertices or do not 
consider labels at all.
The well-known agreement subtree problem, for example, considers the case, where 
only the leaf nodes are labeled, with no label appearing more than once per 
tree~\cite{MaTh13}.
We discuss the approaches most relevant for our work.
Gupta and Nishimura~\cite{GuNi1998} investigated the \emph{largest common embeddable
subtree} problem in unlabeled rooted trees.
Their definition is based on topological embedding (or homeomorphism) and allows 
to map edges of the common subtree to vertex-disjoint paths in the input trees.
The algorithm uses the classical idea to decompose the problem into subproblems 
for smaller trees, which are solved via bipartite matching. A solution for two 
trees with at most $n$ vertices is computed in time $\mathcal{O}(n^{2.5} \log n)$
using a dynamic programming approach.
Fig.~\ref{subfig:mcsles} illustrates the difference between maximum common subgraph and largest common embeddable subtree.
Lozano and Valiente~\cite{LoVa2004} investigated the \emph{maximum common embedded subtree} problem, which is based on edge contraction. In both cases the input graphs are rooted unlabeled trees.
Note, the definition of their problems is not equivalent.
The first is polynomial time solvable, while the second is \textsf{NP}-hard for unordered trees, but polynomial time solvable for ordered trees.
Many algorithms do not support trees, where leaves and the inner vertices both have labels. A notable exception is the approach by Kao et al.~\cite{Kao2001}, where only vertices with the same label may be mapped. This algorithm generalizes the approach by Gupta/Nishimura and improves its running time to 
$\mathcal{O}(\sqrt{d}D \log\frac{2n}{d})$, where $D$ denotes the number
of vertex pairs with the same label and $d$ the maximum degree of all vertices.

We consider the problem of finding a \emph{largest weight common subtree embedding} (\WES),
where matching vertices are not required to have the same label, but their 
degree of agreement is determined by a weight function.
We build on the basic ideas of Gupta and Nishimura~\cite{GuNi1998}.
To prevent arbitrarily long paths which are mapped to a common edge we study 
a linear distance penalty for paths of length 
greater than 1.
Note that, by choosing a high distance 
penalty, we solve the maximum common subtree (MCS) problem as a special case. 
By choosing weight 1 for equal labels and sufficiently small negative weights 
otherwise, we solve a problem equivalent to the one studied by Kao et al.~\cite{Kao2001}.

\subparagraph{Our contribution.}
We propose and analyze algorithms for finding largest weight common subtree 
embeddings. Our method requires to solve a series of bipartite matching instances
as subproblem, which dominates the total running time. 
We build on recent results by Ramshow and Tarjan~\cite{RaTa12} for unbalanced matchings.
Let $T$ and $T'$ be labeled rooted trees with $k:=|T|$ and $l:=|T'|$ vertices, respectively, and $\Delta:=\min\{\Delta(T),\Delta(T')\}$ the smaller degree of the two input trees.
For real-valued weight functions we prove a time bound of $\mathcal O(kl\Delta)$.
For integral weights bounded by a constant $C$ we prove a running time of 
$\mathcal O(kl\sqrt \Delta\log(\min\{k,l\}C))$.
This is an improvement over the algorithm by Kao et al.\,\cite{Kao2001} if there 
are only few labels and the maximum degree of one tree is much smaller than the 
maximum degree of the other.
In addition, we support weights and a linear penalty for skipped vertices.

Moreover, the algorithm by Kao et al.\,\cite{Kao2001} is designed for rooted 
trees only. A straight forward approach to solve the problem for unrooted 
trees is to try out all pairs of possible roots, which results in an additional 
$\mathcal O(kl)$ factor.
However, our algorithm exploits the fact that there are many similar matching instances
using techniques related to~\cite{Chung1987,DrKrMu2016}.
This includes computing additional matchings of cardinality two.
For unrooted trees and real-valued weight functions we prove the same $\mathcal O(kl\Delta)$ time bound as for rooted trees.
This leads to an improvement over the formally best algorithm for solving the maximum 
common subtree problem, for which a time bound  of $\mathcal O(kl\,(\Delta + \log d))$ has been proven~\cite{DrKrMu2016}.

\section{Preliminaries}
\label{sec:prel}
We consider finite simple undirected graphs.
Let $G=(V,E)$ be a graph, we refer to the set of \emph{vertices} $V$ by $V(G)$ 
or $V_G$ and to the set of \emph{edges} by $E(G)$ or $E_G$. An edge connecting 
two vertices $u, v \in V$ is denoted by $uv$ or $vu$. The \emph{order} $|G|$ of 
a graph $G$ is its number of vertices.
The \emph{neighbors} of a vertex $v$ are defined as $N(v):=\{u\in V_G\mid vu\in E_G\}$.
The \emph{degree} of a vertex $v\in V_G$ is $\delta(v):=|N(v)|$, the \emph{degree} $\Delta(G)$ of a graph $G$ is the maximum degree of its vertices.

A \emph{path} $P$ is a sequence of pairwise disjoint vertices connected through edges and denoted as $P=(v_0,e_1,v_1,\ldots, e_l,v_l)$.
We alternatively specify the vertices $(v_0,\ldots,v_l)$ or edges $(e_1,\ldots,e_l)$ only.
The \emph{length} of a path is its number of edges.
A connected graph with a unique path between any two vertices is a \emph{tree}.
A tree $T$ with an explicit root vertex $r\in V(T)$ is called \emph{rooted} tree, denoted by $T^r$.
In a rooted tree $T^r$ we denote the set of \emph{children} of a vertex $v$ by $C(v)$
and its \emph{parent} by $p(v)$, where $p(r) = r$.
For any tree $T$ and two vertices $u,v\in V(T)$ the \emph{rooted subtree} $T^u_v$ is induced by the vertex $v$ and its descendants related to the tree $T^u$.
If the root $r$ is clear from the context, we may abbreviate $T_v:=T^r_v$.
We refer to the root of a rooted tree $T$ by $r(T)$.

If the vertices of a graph $G$ can be separated into exactly two disjoint sets $V,\ U$ such that $E(G)\subseteq V\times U$, then the graph is called \emph{bipartite}. In many cases the disjoint sets are already given as part of the input. In this case we write $G=(V\sqcup U,E)$, where $E\subseteq V\times U$. 
 
For a graph $G=(V,E)$ a \emph{matching} $M\subseteq E$ is a set of edges, such that no two edges share a vertex.
For an edge $uv\in M$, the vertex $u$ is the \emph{partner} of $v$ and vice versa.
The \emph{cardinality of a matching} $M$ is its number of edges $|M|$.
A \emph{weighted graph} is a graph endowed with a function $\wMatch:E\to\mathbb R$.
The weight of a matching $M$ in a weighted graph is $\WMatch(M):=\sum_{e\in M} \wMatch(e)$.
We call a matching $M$ of a weighted graph $G$ a \emph{maximum weight matching} (\MWM) if there is no other matching $M'$ of $G$ with $\WMatch(M')>\WMatch(M)$.
A matching $M$ of $G$ is a \emph{\MWM of cardinality $k$} (\MWMk{k}) if there is no other \MWMk{k} $M'$ of $G$ with $\WMatch(M')>\WMatch(M)$.

For convenience we define the maximum of an empty set as $-\infty$.

\section{Gupta and Nishimura's algorithm}
\label{sec:les}
In this section we formally define a \emph{Largest Common Subtree Embedding} (\LES) and present a brief overview of Gupta and Nishimura's algorithm to compute such an embedding. The following two definitions are based on~\cite{GuNi1998}.
\begin{definition}[Topological Embedding]
\label{def:topemb}
A rooted tree $T$ is \emph{topologically embeddable} in a rooted tree $T'$ if there is an injective function $\psi:V(T)\to V(T')$, such that $\forall a,b,c\in V(T)$
\begin{enumerate}
\item [i)]If $b$ is a child of $a$, then $\psi(b)$ is a descendant of $\psi(a)$.
\item [ii)]For distinct children  $b,c$ of $a$, the paths from $\psi(a)$ to $\psi(b)$ and from $\psi(a)$ to $\psi(c)$ have exactly $\psi(a)$ in common.
\end{enumerate}
$T$ is \emph{root-to-root topologically embeddable} in $T'$, if $\psi(r(T))=r(T')$.
\end{definition}
\begin{definition}[(Largest) Common Subtree Embedding; (L)\ES]
\label{def:les}
Let $T$ and $T'$ be rooted trees and $S$ be topologically embeddable in both $T$ and $T'$.
For such a $S$ let $\psi:V(S)\to V(T)$ and $\psi':V(S)\to V(T')$ be topological embeddings.
\begin{itemize}
\item Then $\varphi:=\psi'\circ\psi^{-1}$ is a \emph{Common Subtree Embedding}.
\item If there is no other tree $S'$ topologically embeddable in both $T$ and $T'$ with $|S'|>|S|$, then $S$ is a \emph{Largest Common Embeddable Subtree} and $\varphi$ is a \emph{Largest Common Subtree Embedding}.
\item An \ES with $\varphi(r(T))=r(T')$ is a \emph{root-to-root \ES}.
\item A root-to-root \ES is \emph{largest}, if it is of largest weight among all root-to-root common subtree embeddings.
\end{itemize}
\end{definition}
 
\subparagraph{Algorithm from Gupta and Nishimura.}
\label{sect:algGN}
Gupta and Nishimura~\cite{GuNi1998} presented an algorithm to compute the size of a largest common embeddable subtree based on dynamic programming, which is similar to the computation of a largest common subtree, described in, e.g.,\cite{Matula1978,DrKrMu2016}.
Let $T$ and $T'$ be rooted trees and $\TabL$ be a table of size $|T||T'|$.
For each pair of vertices $u\in T, v\in T'$ the value $\TabL(u,v)$ stores the size of a \LES between the rooted subtrees $T_u$ and $T'_v$.
Gupta and Nishimura proved, that an entry $\TabL(u,v)$ is determined by the maximum of the following three quantities.
\begin{itemize}
\item $M_1=\max\{\TabL(u,c)\mid c\in C(v)\}$
\item $M_2=\max\{\TabL(b,v)\mid b\in C(u)\}$
\item $M_3=\WMatch(M)+1$, where $M$ is a \MWM of the complete bipartite graph $(C(u)\sqcup C(v),C(u)\times C(v))$ with edge weight $\wMatch(bc)=\TabL(b,c)$ for each pair $(b,c)\in C(u)\times C(v)$
\end{itemize}
Here, $M_1$ represents the case, where the vertex $v$ is not mapped.
To satisfy ii) from Def.\,\ref{def:topemb}, we may map at most one child $c\in C(v)$.
$M_2$ represents the case, where $u$ is not mapped and at most one child $b\in C(u)$ is allowed.
$M_3$ represents the case $\varphi(u)=v$.
To maximize the number of mapped descendants we compute a maximum weight matching, where the children of $u$ and $v$ are the vertex sets $C(u)$ and $C(v)$, respectively, of a bipartite graph.
The edge weights are determined by the previously computed solutions, i.e., the {\LES}s between the children of $u$ and $v$ and their descendants, namely between $T_b$ and $T'_c$ for each pair of children $(b,c)$.
The algorithm proceeds from the leaves to the roots.
From the above recursive formula, we get $\TabL(u,v)=1$ if $u$ or $v$ is a leaf, which was separately defined in~\cite{GuNi1998}.

A maximum value in the table yields the size of an \LES.
We obtain the size of a root-to-root \LES from $M_3$ of the root vertices $r(T),r(T')$.
Note, with storing $\mathcal O(|\TabL|)$ additional data, it is easy to obtain a (root-to-root) \LES $\varphi$.
\begin{theorem}[Gupta, Nishimura, \cite{GuNi1998}]
\label{theorem:gupta}
Computing an \LES between two rooted trees of order at most $n$ is possible in time $\mathcal O(n^{2.5}\log n)$.
\end{theorem}
\section{Largest Weight Common Subtree Embeddings}
First, we introduce weighted common subtree embeddings between labeled trees.
Part of the input is an integral or real weight function on all pairs of the labels.
Next, we consider a linear distance penalty for skipped vertices in the input trees.
After formalizing the problem and presenting an algorithm, we prove new upper time bounds.

\subparagraph{Vertex Labels.}
In many application domains the vertices of the trees need to be distinguished.
A common representation of a vertex labeled tree $T$ is $(T,l)$, where $l:V(T)\to \Sigma$ with $\Sigma$ as a finite set of labels.
Let $\wIso:\Sigma\times \Sigma\to \mathbb{R}\cup \{-\infty\}$ assign a weight to each pair of labels.
Instead of maximizing the number of mapped vertices, we want to maximize the sum of the weights $\wIso(l(u),l(\varphi(u)))$ of all vertices $u$ mapped by a common subtree embedding $\varphi$.
For simplicity, we will omit $l$ and $l'$ for the rest of this paper and define $ \wIso(u,v):=\wIso(l(u),l'(v))$ for any two vertices $(u,v)\in V(T)\times V(T')$.

\subparagraph{Edge Labels.}
Although not as common as vertex labels, edge labels are useful to represent different bonds between atoms or relationship between individuals.
In a common subtree embedding we do not map edges to edges but paths to paths.
Since in an embedding inner vertices on mapped paths do not contribute to the weight, we do the same with edges.
I.e., both paths need to have length 1 for their edge labels to be considered.
Here again, we want to maximize the weight $\wIso(e,e'):=\wIso(l(e),l'(e'))$ of these edges mapped to each other (additional to the weight of the mapped vertices).

\subparagraph{Distance Penalties.}
Depending on the application purpose it might be desirable that paths do not have arbitrary length.
Here, we introduce a linear distance penalty for paths of length greater than 1.
I.e., each inner vertex on a path corresponding to an edge of the common embeddable subtree lowers the weight by a given penalty $p$.
By assigning $p$ the value $\infty$ we effectively compute a maximum common subtree.
The following definition formalizes an \LES under a weight function $\wIso$ and a distance penalty $p$.

\begin{definition}[Largest Weight Common Subtree Embedding; \WES]
\label{def:les_weight_dist}
Let $(T,l)$ and $(T',l')$ be rooted vertex and/or edge labeled trees.
Let $\varphi$ be a common subtree embedding from $T$ to $T'$.
Let $\wIso:\Sigma\times \Sigma\to \mathbb{R}\cup \{-\infty\}$ assign a weight to each pair of labels.
Let $p \in \mathbb{R}^{\geq 0}\cup\{\infty\}$ be a distance penalty.
We refer to a path $P_1=(u_0,e_1,u_1,\ldots,u_k)$ in the tree $T$ corresponding to a single edge in the common embeddable subtree as \emph{topological path}.
Let $\varphi(P_1)$ be the corresponding path $(v_0,e_1',v_1,\ldots,v_l)$ in $T'$. Then
\begin{itemize}
\item $\wPath(P_1,\varphi(P_1))=\wIso(e_1,e_1'):=\wIso(l(e_1),l'(e_1'))$, if $l=k=1$, or
\item $\wPath(P_1,\varphi(P_1))=-p\cdot(l+k-2)$, otherwise.
\item The \emph{weight} $\WIso(\varphi)$ is the sum of the weights $\wIso(u,\varphi(u))$ of all vertices $u$ mapped by $\varphi$ plus the weights $\wPath(P,\varphi(P))$  of all topological paths $P$.
\item If $\varphi$ is of largest weight among all common subtree embeddings, then $\varphi$ is a \emph{Largest Weight Common Subtree Embedding} (\WES).
\end{itemize}
\end{definition}

The definition of root-to-root \WES is analogue to Def.\,\ref{def:les}.
A closer look at the definition of $\wPath$ reveals that each inner vertex (the skipped vertices) on a topological path or its mapped path subtracts $p$ from the embedding's weight.
Fig.\,\ref{subfig:les_edges} illustrates two weighted common subtree embeddings.

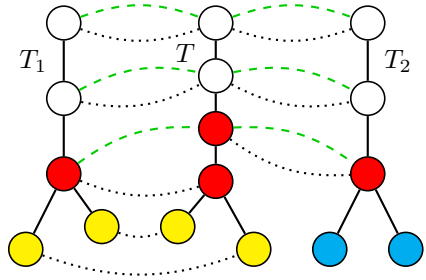
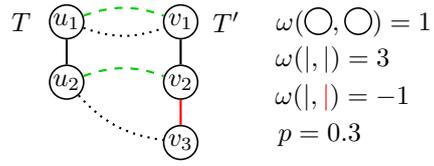
\begin{figure}
\centering
\begin{subfigure}[b]{0.48\textwidth}
\centering
\begin{tikzpicture}%
\node at (-0.4,2.5) {$T_1$};
\node[vertexWhiteLarge] (s1) at (0,3) {};
\node[vertexWhiteLarge] (s2) at (0,2) {};
\node[vertexWhiteLarge,fill=red] (s3) at (0,1) {};
\node[vertexWhiteLarge,fill=yellow] (s4) at (-0.5,0) {};
\node[vertexWhiteLarge,fill=yellow] (s5) at (0.5,0.3) {};

\draw[edge] (s1) -- (s2)--(s3)--(s4);
\draw[edge] (s3)--(s5);

\node at (1.6,2.6) {$T$};
\node[vertexWhiteLarge] (t1) at (2,3) {};
\node[vertexWhiteLarge] (t2) at (2,2.3) {};
\node[vertexWhiteLarge,fill=red] (t2b) at (2,1.6) {};
\node[vertexWhiteLarge,fill=red] (t3) at (2,0.9) {};
\node[vertexWhiteLarge,fill=yellow] (t4) at (1.5,0.3) {};
\node[vertexWhiteLarge,fill=yellow] (t5) at (2.5,0) {};

\draw[edge] (t1) -- (t2)--(t2b)--(t3)--(t4);
\draw[edge] (t3)--(t5);

\node at (4.4,2.5) {$T_2$};
\node[vertexWhiteLarge] (a1) at (4,3) {};
\node[vertexWhiteLarge] (a2) at (4,2) {};
\node[vertexWhiteLarge,fill=red] (a3) at (4,1) {};
\node[vertexWhiteLarge,fill=cyan] (a4) at (3.5,0) {};
\node[vertexWhiteLarge,fill=cyan] (a5) at (4.5,0) {};

\draw[edge] (a1) -- (a2)--(a3)--(a4);
\draw[edge] (a3)--(a5);

\draw[edgeIsoGreen] (s1) edge (t1);
\draw[edgeIsoGreen] (s2) edge (t2);
\draw[edgeIsoGreen] (s3) edge (t2b);
\draw[edgeIsoGreen] (t1) edge (a1);
\draw[edgeIsoGreen] (t2) edge (a2);
\draw[edgeIsoGreen] (t2b) edge (a3);

\draw[edgeIso] (s1) edge (t1);
\draw[edgeIso] (s2) edge (t2);
\draw[edgeIso] (s3) edge (t3);
\draw[edgeIso] (s4) edge (t5);
\draw[edgeIso] (s5) edge (t4);
\draw[edgeIso] (t1) edge (a1);
\draw[edgeIso] (t2) edge (a2);
\draw[edgeIso] (t2b) edge (a3);
\end{tikzpicture}
\subcaption{Labeled MCS (green, dashed) and \LES (black, dotted) between $T$ and $T_i$, $i\in \{1,2\}$.}
\label{subfig:mcsles}
\end{subfigure}
\ \ 
\begin{subfigure}[b]{0.49\textwidth}
\centering
\begin{tikzpicture}
\node at (4.78,4.5) {$\wIso(\tikz[baseline=-0.75ex]{\node[vertexWhite] {};}
,\tikz[baseline=-0.75ex]{\node[vertexWhite] {};})=1$};
\node at (4.5,4) {$\wIso(|,|)=3$};
\node at (4.63,3.5) {$\wIso(|,\textcolor{red}{|})=-1$};
\node at (4.33,3) {$p=0.3$};

\node at (0.4,4.5) {$T$};
\node[vertexWhiteLarge] (ta) at (1.0,4.5) {};
\node at (ta) {$u_1$};
\node[vertexWhiteLarge] (tb) at (1.0,3.7) {};
\node at (tb) {$u_2$};
\draw[edge] (ta) -- (tb);

\node at (3.1,4.5) {$T'$};
\node[vertexWhiteLarge] (s1) at (2.5,4.5) {};
\node at (s1) {$v_1$};
\node[vertexWhiteLarge] (s2) at (2.5,3.7) {};
\node at (s2) {$v_2$};
\node[vertexWhiteLarge] (s3) at (2.5,2.9) {};
\node at (s3) {$v_3$};
\draw[edge] (s1) -- (s2);
\draw[edge,draw=red] (s2) -- (s3);
\draw[edgeIso] (ta) edge (s1);
\draw[edgeIsoGreen] (ta) edge (s1);
\draw[edgeIso] (tb) edge (s3);
\draw[edgeIsoGreen] (tb) edge (s2);
\end{tikzpicture}
\subcaption{The black embedding has weight 1.7, since the vertex $v_2$ is skipped and therefore the penalty $p$ is applied; the weight between the edges is not added. The green embedding has weight 5, 2 from the vertices, 3 from the topological paths of length 1.}
\label{subfig:les_edges}
\end{subfigure}
\caption{\textbf{\subref{subfig:mcsles})} Although 'intuitively' $T$ is more similar to $T_1$ than to $T_2$, both MCSs have size 3. However, the \LES between $T$ and $T_1$ has 6 mapped vertices. \textbf{\subref{subfig:les_edges})} Two weighted embeddings; one with a skipped vertex, the other where the edge labels contribute to the weight.}
\label{fig:mcs_edge}
\end{figure}

\subparagraph{The dynamic programming approach.}
To compute a \WES, we need to store some additional data during the computation.
In Gupta and Nishimura's algorithm there is a table $\TabL$ of size $|T||T'|$ to store the weight of \LES{}s between subtrees of the input trees.
In our algorithm we need a table $\TabL$ of size $2|T||T'|$.
An entry $\TabL(u,v,t)$ stores the weight of a \WES between the rooted subtrees $T_u$ and $T'_v$ of \emph{type} $t\in\{\trtr,\tsk\}$.
Type \trtr represents a \emph{root-to-root} embedding between $T_u$ and $T'_v$; \tsk an embedding, where $u$ or $v$ is \emph{skipped}.
Skipped in the sense, that at least on of $u,v$ will be an inner vertex when mapping some ancestor nodes during the dynamic programming.
For type \tsk we subtract the penalty $p$ from the weight for the skipped vertices before storing it in our table.
We obtain the weight of a \WES and a root-to-root \WES, respectively, from the maximum value of type \trtr and from $\TabL(r(T),r(T'),\trtr)$, respectively.
The following lemma specifies the recursive computation of an entry $L(u,v,t)$.
\begin{lemma}
\label{lemma:recursion}
Let $u\in V(T)$ and $v\in V(T')$. 
For $\tin$ let $M^T_{t}=\max\{\TabL(b,v,t)\mid b\in C(u)\}$ and $M^{T'}_{t}=\max\{\TabL(u,c,t)\mid c\in C(v)\}$. Then
\begin{itemize}
\item $\TabL(u,v,\tsk)=\max\{M^T_{\tsk},M^T_{\trtr},M^{T'}_{\tsk},M^{T'}_{\trtr}\}-p$
\end{itemize}
Let $G=(C(u)\sqcup C(v),C(u)\times C(v))$ be a bipartite graph with edge weights $\wMatch(bc)=\max\{\TabL(b,c,\tsk),\TabL(b,c,\trtr)+\wIso(ub,vc)\}$  for each pair $(b,c)\in C(u)\times C(v)$. Then
\begin{itemize}
\item $\TabL(u,v,\trtr)=\wIso(u,v)+\WMatch(M)$, where $M$ is a \MWM on $G$.
\end{itemize}
\end{lemma}

\begin{proof}
The type $\tsk$ represents the case of an embedding between $T_u$ and $T'_v$ which is not root-to-root.
From the definition of $M^T_t$ the vertex $u$ is skipped and from the definition of $M^{T'}_t$ the vertex $v$ is skipped, so it is indeed not root-to-root.
Since either $u$ or $v$ was skipped, we subtract the penalty $p$.
This ensures we have taken inner vertices of later steps of the dynamic programming into account.

The type $\trtr$  implies that $u$ is mapped to $v$.
Each edge in $G$ represents the weight of a \WES from one child of $u$ to one child of $v$.
A maximum matching yields the best combination which satisfies Def.\,\ref{def:topemb}.
If a child $b$ of $u$ is mapped to a child $c$ of $v$, the paths $bu$ and $cv$ have length one.
Then from Def.\,\ref{def:les_weight_dist} we have to add the weight $\wIso(bu,cv)$.
Otherwise at least one path has length greater than one and we have to subtract the distance penalty $p$ for each inner vertex.
We already did that while computing $\TabL(b,c,\tsk)$.
\end{proof}

\subparagraph{Time and space complexity.}
We next analyze upper time and space bounds.
Thereby we distinguish between real- and integer-valued weight functions $\wIso$.
If we use dynamic programming starting from the leaves to the roots, we need to compute each value $\TabL(u,v,t)$ only once.

\begin{theorem}
\label{theorem:TimeSpaceRootedLES}
Let $T$ and $T'$ be rooted vertex and/or edge labeled trees.
Let $\wIso$ be a weight function, $\Delta=\min \{\Delta(T),\Delta(T')\}$, and $p$ be a distance penalty.
\begin{itemize}
\item A \WES between $T$ and $T'$ can be computed in time $\mathcal{O}(|T|\,|T'|\Delta\}$ and space $\mathcal O(|T|\,|T'|)$.
\item If the weights are integral and bounded by a constant $C$, a \WES can be computed in time $\mathcal O(|T|\,|T'|\sqrt\Delta\log (C\min\{|T|,|T'|\}))$.
\end{itemize}
\end{theorem}

We first need to provide two results regarding maximum weighted matchings.

\begin{lemma}[\cite{RaTa12_TR}]
\label{lemma:matching_real_single}
Let $G$ be a weighted bipartite graph containing $m$ edges and vertex sets of sizes $s$ and $t$.
W.l.o.g. $s\leq t$.
We can compute a \MWM on $G$ in time $\mathcal O(ms+s^2\log s)$ and space $\mathcal O(m)$.
If $G$ is complete bipartite, we may simplify the time bound to $\mathcal O(s^2t)$.
\end{lemma}

If the weights are integral and bounded by a constant, the following result for a minimum weight matching was shown in~\cite{Goldberg2017}.
We solve \MWM by multiplying the weights by -1.
\begin{lemma}
\label{lemma:matching_int_single}
Let $G$ as in Lemma~\ref{lemma:matching_real_single} and the weights be integral and bounded by a constant $C$.
We can compute a \MWM on $G$ in time $\mathcal O(m\sqrt{s}\log C)$.
If $G$ is complete bipartite, we may simplify the time bound to $\mathcal O(s^{1.5}t\log C)$.
\end{lemma}
Unfortunately, there is no space bound given in~\cite{Goldberg2017}.
However, since their algorithm is based on flows, the space bound is probably $\mathcal O(m)$.
If we assume this to be correct, the space bound in Theorem~\ref{theorem:TimeSpaceRootedLES} applies to integral weights too.

\begin{proof}[Proof of Theorem~\ref{theorem:TimeSpaceRootedLES}]
We observe that the entries of type $\trtr$ in the table $\TabL$ dominate the computation time.
We assume the bipartite graphs on which we compute the \MWM{}s to be complete.
We further observe that each edge $bc$ representing the weight of the \WES between the subtrees $T_b$ and $T'_c$ is contained in exactly one of the matching graphs.

Let us assume real weights first.
$\TabL$ requires $\mathcal O(|T|\,|T'|)$ space.
We may also compute each \MWM within the same space bound.
This proves the total space bound.
From Lemma~\ref{lemma:matching_real_single} the time to compute all the \MWM{}s is bounded by 
\begin{align*}
&\mathcal O\left(\sum_{v\in V_{T'}}\sum_{u\in V_T} |C(u)|\,|C(v)| \min \{|C(u)|,|C(v)|\}\right)\\
\subseteq\ & \mathcal O\left(\sum_{v\in V_{T'}} |C(v)| \sum_{u\in V_T} |C(u)|\Delta\right)
=\mathcal O(|T|\,|T'|\Delta). 
\end{align*}

Let us assume $\wIso$ to be integral and bounded by a constant $C$ next.
This implies a weight of each single matching edge of at most $\bar C:=2C\cdot\min\{|T|,|T'|\}$, since no more than $2C-1$ edges and vertices in total may contribute to the weight.
Negative weight edges never contribute to a \MWM and may safely be omitted.
From Lemma~\ref{lemma:matching_int_single} the time bound is
\begin{align*}
& \mathcal O\left(\sum_{v\in V_{T'}}\sum_{u\in V_T} |C(u)|\,|C(v)|\sqrt{\min \{|C(u)|,|C(v)|\}}\log \bar{C}\right)\\
\subseteq\ & \mathcal O\left(\sum_{v\in V_{T'}} |C(v)| \sum_{u\in V_T} |C(u)|\sqrt\Delta\log \bar{C}\right)
=\mathcal O\left(|T|\,|T'|\sqrt\Delta\log (C\min\{|T|,|T'|\})\right).
\end{align*}
\end{proof}

\section{Largest Weight Common Subtree Embeddings for Unrooted Trees}
\label{sec:unrooted}
In this section we consider a \WES between unrooted trees.
I.e., we want to find two root vertices $r\in V(T)$, $s\in V(T')$ and a common subtree embedding $\varphi$ between $T^r$ and $T'^{s}$ such that there is no embedding $\varphi'$ between $T^{r'}$ and $T'^{s'}$, $r'\in V(T)$, $s'\in V(T')$, with $\WIso(\varphi')>\WIso(\varphi)$.
We abbreviate this as \WESu.
In Sect.\,\ref{sec:unrooted_basic_fixing} we present a basic algorithm and a first improvement by fixing the root of $T$.
In Sect.\,\ref{sec:unrooted_similarities} we speed up the computation by exploiting similarities between the different chosen roots of $T'$.
In each section we prove the correctness and upper time bounds of our algorithms.

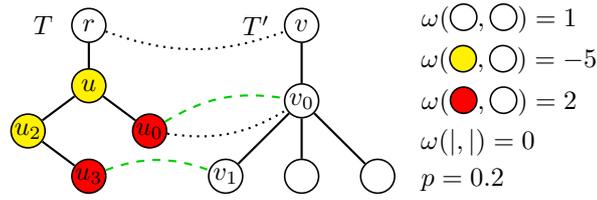
\begin{figure}
\center
\begin{tikzpicture}
\node at (7.1,4.6) {$\wIso(\tikz[baseline=-0.75ex]{\node[vertexWhite] {};}
,\tikz[baseline=-0.75ex]{\node[vertexWhite] {};})=1$};
\node at (7.23,4.05) {$\wIso(\tikz[baseline=-0.75ex]{\node[vertexWhite,fill=yellow] {};}
,\tikz[baseline=-0.75ex]{\node[vertexWhite] {};})=-5$};
\node at (7.1,3.5) {$\wIso(\tikz[baseline=-0.75ex]{\node[vertexWhite,fill=red] {};}
,\tikz[baseline=-0.75ex]{\node[vertexWhite] {};})=2$};
\node at (6.82,2.95) {$\wIso(|,|)=0$};
\node at (6.62,2.45) {$p=0.2$};

\node at (1.1,4.5) {$T$};
\node[vertexWhiteLarge] (ta) at (1.7,4.5) {};
\node at (ta) {$r$};
\node[vertexWhiteLarge, fill=yellow] (tb) at (1.7,3.7) {};
\node at (tb) {$u$};
\node[vertexWhiteLarge, fill=yellow] (tc) at (0.9,3.1) {};
\node at (tc) {$u_2$};
\node[vertexWhiteLarge, fill=red] (td) at (1.7,2.5) {};
\node at (td) {$u_3$};
\node[vertexWhiteLarge, fill=red] (te) at (2.5,3.1) {};
\node at (te) {$u_0$};

\draw[edge] (ta) -- (tb);
\draw[edge] (tb) -- (tc);
\draw[edge] (tc) -- (td);
\draw[edge] (tb) -- (te);

\node at (3.9,4.5) {$T'$};

\node[vertexWhiteLarge] (s1) at (4.5,4.5) {};
\node at (s1) {$v$};
\node[vertexWhiteLarge] (s2) at (4.5,3.5) {};
\node at (s2) {$v_0$};
\node[vertexWhiteLarge] (s3) at (3.5,2.5) {};
\node at (s3) {$v_1$};
\node[vertexWhiteLarge] (s4) at (4.5,2.5) {};
\node at (s4) {};
\node[vertexWhiteLarge] (s5) at (5.5,2.5) {};
\node at (s5) {};

\draw[edge] (s1) -- (s2);
\draw[edge] (s2) -- (s3);
\draw[edge] (s2) -- (s4);
\draw[edge] (s2) -- (s5);

\draw[edgeIso] (ta) edge (s1);
\draw[edgeIso] (te) edge (s2);
\draw[edgeIsoGreen] (te) edge (s2);
\draw[edgeIsoGreen] (td) edge (s3);
\end{tikzpicture}
\caption{The weight of a \WES between $T^{r}$ and $T'^{v}$ is 2.8 (black, dotted), since we have one skipped vertex $u$ for penalty $0.2$. The weight of a \WES between $T^{u_0}$ and $T'^{v_0}$ is 3.6 (green, dashed) for two skipped vertices. The latter one is also a \WESu.}
\label{fig:lesu}
\end{figure}

\subsection{Basic algorithm and fixing one root}
\label{sec:unrooted_basic_fixing}
The basic idea is to compute for each pair of vertices $(u,v)\in V(T)\times V(T')$ a (rooted) \WES from $T^u$ to $T'^v$ and output a maximum solution.
This is obviously correct and the time bound is $\mathcal O(|T|^2\,|T'|^2\Delta)$.

In our previous work~\cite{DrKrMu2016} we showed how to compute a maximum common subtree between unrooted trees by arbitrarily choosing one root vertex $r$ of $T$ and then computing MCSs between $T^r$ and $T'^s$ for all $s\in V(T')$.
We may adapt this strategy to find a \WESu between the input trees.
However, Fig.\,\ref{fig:lesu} shows that this strategy sometimes fails. 
A \WES between $T^r$ and $T'$ rooted at any vertex results in a weight of at most 2.8.
However, rooting $T$ at $u_0$ results in a \WES of weight 3.6.

In a maximum common subtree between trees $T$ and $T'$, let $r \in V(T)$ be an arbitrarily chosen root of $T$.
If any two children $u_1,u_2$ of their parent node $u\in V(T)$ are mapped to vertices of $T'$, then $u$ is also mapped.
This statement is independent from the chosen root node, since a common subtree is connected.
If we want to compute a (rooted) \WES, the statement is also true (for the given root).
This follows from Def.\,\ref{def:topemb} ii).
However, if we choose $u_1$ as root in a \WESu, we may skip $u$ and map $u_2$, forming the topological path $(u_1,u,u_2)$.
Whatever we do, if we skip vertex $u$ as an inner vertex of a topological path, this is the only path containing $u$; otherwise we violate Def.\,\ref{def:topemb}. We record this as a lemma.
\begin{lemma}
Let $T$ and $T'$ be unrooted trees.
Let $\varphi$ be a \WESu from $T$ to $T'$ and $u\in V(T)$ be an inner vertex of a topological path with its neighbors $N(u)=\{u_1,u_2,\ldots,u_k\}$.
Then $\varphi$ maps vertices from exactly two of the rooted subtrees $T^u_{u_1},\ldots,T^u_{u_k}$ to $T'$.
\end{lemma}

To compute a \WESu $\varphi$, additionally to the strategy from~\cite{DrKrMu2016}, we need to cover the case, that there is no single vertex $u$ mapped by $\varphi$, such that all vertices mapped by $\varphi$ are contained in $T^r_u$, cf.\,Fig.\,\ref{fig:lesu} with $r$ as chosen root.
In this case let $u$ be the unique inner vertex of a topological path $P$, such that all vertices mapped by $\varphi$ are contained in $T^r_u$.
An example is the yellow vertex $u$ in Fig.\,\ref{fig:lesu}.
Further, let $P_1=(u_0,\ldots,u_{i-1},u_i=u,u_{i+1},\ldots,u_k)$ be the topological path containing $u$ and $\varphi(P_1)=P_2=(v_0,\ldots,v_l)$ with $\varphi(u_0)=v_0$ and $\varphi(u_k)=v_l$.

Then there is a \WES $\phi_1$ between the rooted subtrees $T^r_{u_{i-1}}$ and $T'^{v_1}_{v_0}$
containing $u_0$ and all its descendants mapped by $\varphi$.
There is another \WES $\phi_2$ between the rooted subtrees $T^r_{u_{i+1}}$ and $T'^{v_0}_{v_1}$ containing $u_k$ and all its descendants mapped by $\varphi$.

For any vertex $v\in V(T')$ let $\TabL^v$ refer to the table $\TabL$ corresponding to $T^r$ and $T'^v$.
Then $L_1:=\max\{\TabL^{v_1}(u_{i-1},v_0,t)\mid \tin\}$ is the weight of the \WES $\phi_1$ minus the penalty for the inner vertices $u_1,\ldots,u_{i-1}$; the penalty is 0, if $i=1$.
$L_2:=\max\{\TabL^{v_0}(u_{i+1},v_1,t)\mid \tin\}$ is the weight of the \WES $\phi_2$ minus the penalty for the inner vertices $u_{i+1},\ldots,u_{k-1}$ and $v_1,\ldots,v_{l-1}$.
I.e., the penalty $p$ for each inner vertex on the paths excluding $u$ is included in $L_1+L_2$.
Therefore $\WIso(\varphi)=L_1+L_2-p$.
Before summarizing this strategy in the following Lemma~\ref{lemma:wes_strat}, we exemplify it on Fig.\,\ref{fig:lesu}.

The \WESu $\varphi$ is depicted by green dashed lines with mapping $u_0\mapsto v_0$ and $u_3\mapsto v_1$.
The yellow inner vertex $u$ fulfills the condition that $T^r_u$ contains all vertices mapped by $\varphi$.
We have paths $P_1=(u_0,u,u_2,u_3)$ and $\varphi(P_1)=P_2=(v_0,v_1)$.
Then $L_1=\TabL^{v_1}(u_0,v_0,\trtr)=2$ for the mapping $u_0\mapsto v_0$.
Further $L_2=\TabL^{v_0}(u_2,v_1,\tsk)=1.8$ for the mapping $u_3\mapsto v_1$ and the skipped vertex $u_2$.
Note, there is no inner vertex in $P_1$.
We obtain $\WIso(\varphi)=L_1+L_2-p=2+1.8-0.2=3.6$.

\begin{lemma}
\label{lemma:wes_strat}
Let $T$ and $T'$ be trees.
Let $r\in V(T)$ be arbitrarily chosen.
Let $\mathcal W(r,v)$ be the weight of a \WES from $T^r$ to $T'^v$ and
$\mathcal T(u,v,w)=\max\{\TabL^{w}(u,v,t)\mid \tin\}$.
Then the weight of a \WESu is the maximum of the following two quantities.
\begin{itemize}
\item $M_1=\max\{\mathcal W(r,v)\mid v\in V(T')\}$
\item $M_2=\max_{u\in V(T),\ vw\in E(T')}\{\mathcal T(u_1,v,w)+\mathcal T(u_2,w,v)\mid u_1\neq u_2\in C(u)\}-p$
\end{itemize}
\end{lemma}

\begin{lemma}
\label{lemma:wes_strat_complexity}
Let the preconditions be as in Lemma~\ref{lemma:wes_strat}.
Then $M_1$ can be computed in time $\mathcal O(|T|\,|T'|^2\Delta)$ and $M_2$ in time $\mathcal O(|T|\,|T'|)$; both can be computed in space $\mathcal O(|T|\,|T'|)$.
\end{lemma}
\begin{proof}
For any vertices $s,v\in V(T')$ consider the rooted subtree $T'^s_v$. Let $p(v)$ be the parent vertex of $v$ with $p(s)=s$.
Then $T'^s_v=T'^{p(v)}_v=T'^w_v$ for each vertex $w\in V(T')\setminus V(T'^s_v)$, i.e., $w$ is a vertex of $T'$, which is not contained in the rooted subtree $T'^s_v$.
Therefore we can identify each table entry $\TabL^s(u,v,t)$ by $\TabL^{p(v)}(u,v,t)$.
In other words, all the table entries needed to compute $M_1$ are determined first by a node $u\in V(T)$, and second by either an edge $wv\in E(T')$ or the root vertex $s$.
Therefore, the space needed to store all the table entries and thus compute $M_1$ is $\mathcal O(|T|\,|T'|)$.

From Theorem~\ref{theorem:TimeSpaceRootedLES} for each $v\in V(T')$ we can compute $\mathcal W(r,v)$ in time $\mathcal O(|T|\,|T'|\Delta)$.
Thus, the time for $M_1$ is bounded by $\mathcal O(|T|\,|T'|^2\Delta)$.
For any edge $vw\in E(T')$ we observe that the only rooted subtrees from $T'$ to consider are $T'^v_w$ and $T'^w_v$.
Let $L(b,v)$ and $L(b,w)$, $b\in C(u)$ be the weight of a \WES from $T^r_b$ to $T'^w_v$ and $T'^v_w$, respectively.
Let $G$ be a bipartite graph with vertices $C(u)\sqcup \{v,w\}$ and edges between these vertices with weights defined by $L(b,v)$ and $L(b,w)$, respectively.
Let $M$ be a \MWMk{2} on $G$. 
Then $\WMatch(M)=\max\{\mathcal T(u_1,v,w)+\mathcal T(u_2,w,v)\mid u_1\neq u_2\in C(u)\}$.
This follows from the construction of $G$.
Note, a \MWMk{2} contains exactly 2 edges.
Let $C(u)=\{b_1,\ldots,b_k\}$ such that $L(b_i,v)\geq L(b_{i+i},v)$ for any $i<k$.
I.e., the vertices $b_i$ are ordered, such that $L(b_1,v)$ and $L(b_2,v)$ have weight at least $L(b_i,v)$ for all $i>2$.
We remove all edges incident to $v$ except $b_1v$ and $b_2v$.
Analog we remove all but the two edges of greatest weight incident to $w$.
Let $G'$ be the graph with those edges removed and $M'$ be a \MWMk{2} on $G'$.
We next prove $\WMatch(M')=\WMatch(M)$.
Let $M$ be a matching on $G$. Assume the partner of $v$ is $b_i$, $i>2$.
Then let $b=b_1$ if $b_1$ is not the partner of $w$, and $b=b_2$ otherwise.
Replacing $vb_i\in M$ by $vb$  results in a matching $M'$ such that $\WMatch(M')\geq \WMatch(M)$.
We may argue analog for $w$.

Since $G'$ contains at most 4 edges we may compute $M'$ in constant time.
The time to remove the edges from $G$ to $G'$ is $\mathcal O(k)$.
Therefore the time to compute $\max\{\mathcal T(u_1,v,w)+\mathcal T(u_2,w,v)\mid u_1\neq u_2\in C(u)\}$ for given $u$ and $vw$ is $\mathcal O(|C(u)|)$.
The time to compute $M_2$ is $\mathcal O(\sum_{u\in V_T,vw\in E_{T'}} |C(u)|)=\mathcal O(|T|\,|T'|)$.

We may compute $M_2$ from $\TabL$ and additional space $\mathcal O(|T|)$, which is $\mathcal O(|T|\,|T'|)$ in total.
\end{proof}

\subsection{Exploiting similarities}
\label{sec:unrooted_similarities}
In this section we improve the running time from $\mathcal O(|T|\,|T'|^2\Delta)$ to $\mathcal O(|T|\,|T'|\Delta)$. 
To this end, we need to speed up the computation in Lemma~\ref{lemma:recursion}.
Specifically, we exploit similarities between the graphs on which we compute the maximum weight matchings. We further need to speed up the computation of $M^{T'}_{t}$  related to the root vertices from $T'$.
We have to take special care of the sequence, in which we compute the table entries, to avoid circular dependencies. 

\subparagraph{Speeding up the dynamic programming approach.}
In Lemma~\ref{lemma:recursion} the recursion computes maximum values among certain table entries.
We first include the current root $s\in V(T')$ into the notation.
We use the definition of $\TabL^s$ from Sect.\,\ref{sec:unrooted_basic_fixing} refering to the table where $s$ is the root of $V(T')$.
Let $u\in V(T)$ and $v\in V(T')$ be the vertices in the current recursion of Lemma~\ref{lemma:recursion}.
For all $\tin$ let $M^{T,s}_{t}=\max\{\TabL^s(b,v,t)\mid b\in C(u)\}$ and $M^{T',s}_{t}=\max\{\TabL^s(u,c,t)\mid c\in C(v)\}$.
From the proof of Lemma~\ref{lemma:wes_strat_complexity} we know that $T'^s_v=T'^w_v$ for each vertex $w \in V(T')\setminus V(T'^s_v)$.
This implies $M^{T,s}_{t}=M^{T,w}_{t}$ and $M^{T',s}_{t}=M^{T',w}_{t}$ for all $w$ as before.
I.e., it is sufficient to distinguish all the $M^{T,s}_{t}$ and $M^{T',s}_{t}$ first by a node $u\in V(T)$, and second by either an edge $wv\in E(T')$ or a single vertex $s\in V(T')$.

This observation allows us to upper bound the time to compute all the $M^{T,s}_{t}$ by 
\begin{equation*}
\mathcal O\left(\sum_{u\in V_T,wv\in E_{T'}}|C(u)|+\sum_{u\in V_T,s\in V_{T'}}|C(u)|\right)=
\mathcal O\left(\sum_{wv\in E_{T'}}|T|+\sum_{s\in V_{T'}}|T|\right)=
\mathcal O\left(|T|\,|T'|\right). 
\end{equation*}

Let $N(v)=\{c_1,c_2,\ldots c_l\}$ and $C_i:=\{c_1,\ldots,c_{i-1},c_{i+1},\ldots,c_l\}$ for $1\leq i \leq l$, i.e, $C_i$ contains all the vertices from $N(v)$ except $c_i$.
We observe $M^{T',v}_{t}=\max\{\TabL^v(u,c,t)\mid c\in N(v)\}$ and $M^{T',c_i}_{t}=\max\{\TabL^{c_i}(u,c,t)\mid c\in C_i\}
=\max\{\TabL^{v}(u,c,t)\mid c\in C_i\}$ for each $i\in \{1,\ldots,l\}$.
Let $j$ be an index, such that $M^{T',v}_{t}=\TabL^v(u,c_j,t)$.
Then for each $i\neq j$ we have $M^{T',c_i}_{t}=M^{T',v}_{t}$.
Therefore, we may compute $M^{T',v}_{t}$ and $M^{T',{c_i}}_{t}$ for all $i\in \{1,\ldots,l\}$ in time $\mathcal O(\delta(v))$.
Hence, the time to compute all the $M^{T',s}_{t}$ is bounded by
$\mathcal O\left(\sum_{u\in V_T,v\in V_{T'}}\delta(v)\right)= \mathcal O\left(\sum_{u\in V_T}|T'|\right) = \mathcal O(|T|\,|T'|)$.
\begin{lemma}
\label{lemma:speedupM}
Assume there is a sequence of all pairs $(u,v)$ such that all necessary values are available to compute $M^{T,s}_{t}$ and $M^{T',s}_{t}$ for $s\in N(v)\cup \{v\}$; then the total time is  $\mathcal O(|T|\,|T'|)$.
\end{lemma}

\subparagraph{Exploiting similarities between the matching graphs.}
In Lemma~\ref{lemma:recursion} we need to compute a \MWM for each $(u,v)\in V(T)\times V(T')$.
When considering all roots $s\in V(T')$, we have one matching graph $G$ with vertices $C(u)\sqcup N(v)$, $N(v)=\{c_1,\ldots,c_l\}$ as well as $l$ graphs $G_{c_i}$, $1\leq i\leq l$.
This follows analog to the observation regarding $M^{T',v}_{t}$ from the previous paragraph.
A graph $G_{c}$, $c\in N(v)$, is the same as $G$ except that the vertex $c$ and incident edges are removed.
Let $s:=\min\{\delta(u),\delta(v)\}$ and $t:=\max\{\delta(u),\delta(v)\}$.
We now prove a total time bound of $\mathcal O(s^2t)$ for computing a \MWM on $G$ as well as on $G_{c}$ for all $c\in N(v)$.
We distinguish two cases.

i) $s\geq\log t$. In our previous work we presented an algorithm to compute a weighted maximum common subtree, a special case of \WESu ~\cite{DrKrMu2016}. 
A subproblem is to compute \MWM{}s on graphs structurally identical to $G$ and $G_{c_i}$, $1\leq i\leq l$.
We showed that we can compute all those \MWM{}s in time $\mathcal O(st(s+\log t))$.
Under the premise $s\geq\log t$ the time bound is $\mathcal O(s^2t)$.

ii) $s<\log t$. Then one vertex set is much smaller than the other. From Lemma~\ref{lemma:speedup_Matching} we can compute all those \MWM{}s in time $\mathcal O(s^4+s^2t)$.
Under the premise $s<\log t$ that is $\mathcal O(s^2t)$.

\begin{lemma}
\label{lemma:speedupM_MatchingAll}
Assume there is a sequence of all pairs $(u,v)$ such that all necessary values are available to compute the \MWM{s}; then the total time is  $\mathcal O(|T|\,|T'|\Delta)$.
\end{lemma}
\begin{proof}
The time to compute all the \MWM{}s is 
$\mathcal O\left(\sum_{u\in V_T}\sum_{v\in V_{T'}}\delta(u)\delta(v)\min\{\delta(u),\delta(v)\}\right)\subseteq \mathcal O\left(\sum_{u\in V_T}\sum_{v\in V_{T'}}\delta(u)\delta(v)\Delta\right)=\mathcal O(|T|\,|T'|\Delta).$
\end{proof}

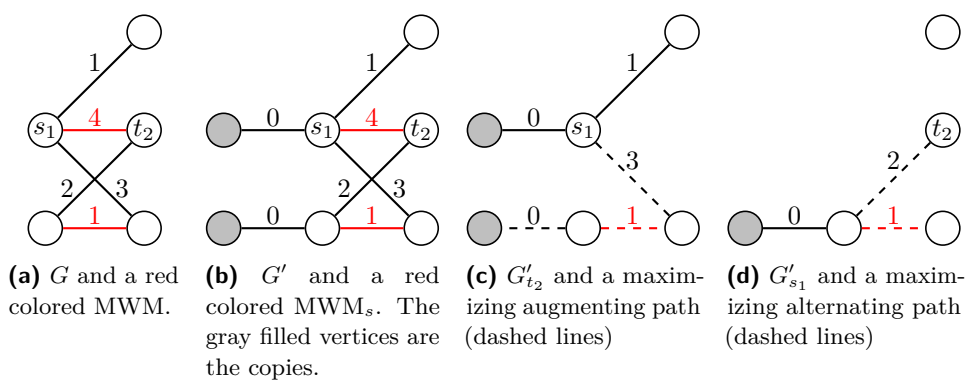
\begin{figure}
\centering
\begin{subfigure}[t]{0.16\textwidth}
\centering
\begin{tikzpicture}%
\node[vertexWhiteLarge] (s1) at (1.3,1.3) {};
\node at (s1) {$s_1$};
\node[vertexWhiteLarge] (s2) at (1.3,0) {};
\node[vertexWhiteLarge] (r1) at (2.6,2.6) {};
\node[vertexWhiteLarge] (r2) at (2.6,1.3) {};
\node at (r2) {$t_2$};
\node[vertexWhiteLarge] (r3) at (2.6,0) {};

\draw[edge] (s1) -- (r1) node [midway, above=0pt] {1};
\draw[edge,draw=red] (s1) -- (r2) node [midway, above=-2pt] {\textcolor{red}4};
\draw[edge] (s1) -- (r3) node [very near end, above=0pt] {3};
\draw[edge] (s2) -- (r2) node [very near start, above=0pt] {2};
\draw[edge,draw=red]  (s2) -- (r3) node [midway, above=-2pt] {\textcolor{red}1};
\end{tikzpicture}
\subcaption{$G$ and a red colored \MWM.}
\label{subfig:G}
\end{subfigure}
\ \ 
\begin{subfigure}[t]{0.22\textwidth}
\begin{tikzpicture}%
\node[vertexWhiteLarge] (s1) at (1.3,1.3) {};
\node at (s1) {$s_1$};
\node[vertexWhiteLarge] (s2) at (1.3,0) {};
\node[vertexWhiteLarge,fill=lightgray] (s1c) at (0,1.3) {};
\node[vertexWhiteLarge,fill=lightgray] (s2c) at (0,0) {};
\node[vertexWhiteLarge] (r1) at (2.6,2.6) {};
\node[vertexWhiteLarge] (r2) at (2.6,1.3) {};
\node at (r2) {$t_2$};
\node[vertexWhiteLarge] (r3) at (2.6,0) {};

\draw[edge] (s1) -- (s1c) node [midway, above=-2pt] {0};
\draw[edge] (s2) -- (s2c) node [midway, above=-2pt] {0};
\draw[edge] (s1) -- (r1) node [midway, above=0pt] {1};
\draw[edge,draw=red] (s1) -- (r2) node [midway, above=-2pt] {\textcolor{red}4};
\draw[edge] (s1) -- (r3) node [very near end, above=0pt] {3};
\draw[edge] (s2) -- (r2) node [very near start, above=0pt] {2};
\draw[edge,draw=red]  (s2) -- (r3) node [midway, above=-2pt] {\textcolor{red}1};
\end{tikzpicture}
\subcaption{$G'$ and a red colored \MWMk{s}. The gray filled vertices are the copies.}
\label{subfig:G'}
\end{subfigure}
\ \ 
\begin{subfigure}[t]{0.22\textwidth}
\begin{tikzpicture}%
\node[vertexWhiteLarge] (s1) at (1.3,1.3) {};
\node at (s1) {$s_1$};
\node[vertexWhiteLarge] (s2) at (1.3,0) {};
\node[vertexWhiteLarge,fill=lightgray] (s1c) at (0,1.3) {};
\node[vertexWhiteLarge,fill=lightgray] (s2c) at (0,0) {};
\node[vertexWhiteLarge] (r1) at (2.6,2.6) {};
\node[vertexWhiteLarge] (r3) at (2.6,0) {};

\draw[edge] (s1) -- (s1c) node [midway, above=-2pt] {0};
\draw[edge,dashed] (s2) -- (s2c) node [midway, above=-2pt] {0};
\draw[edge](s1) -- (r1) node [midway, above=0pt] {1};
\draw[edge,dashed] (s1) -- (r3) node [midway, above=0pt] {3};
\draw[edge,draw=red,dashed]  (s2) -- (r3) node [midway, above=-2pt] {\textcolor{red}1};
\end{tikzpicture}
\subcaption{$G'_{t_2}$ and a maximizing augmenting path (dashed lines)}
\label{subfig:minus_t}
\end{subfigure}
\ \ 
\begin{subfigure}[t]{0.22\textwidth}
\begin{tikzpicture}%
\node[vertexWhiteLarge] (s2) at (1.3,0) {};
\node[vertexWhiteLarge,fill=lightgray] (s2c) at (0,0) {};
\node[vertexWhiteLarge] (r1) at (2.6,2.6) {};
\node[vertexWhiteLarge] (r2) at (2.6,1.3) {};
\node at (r2) {$t_2$};
\node[vertexWhiteLarge] (r3) at (2.6,0) {};

\draw[edge] (s2) -- (s2c) node [midway, above=-2pt] {0};
\draw[edge,draw=red,dashed]  (s2) -- (r3) node [midway, above=-2pt] {\textcolor{red}1};
\draw[edge,dashed] (s2) -- (r2) node [midway, above=0pt] {2};

\end{tikzpicture}
\subcaption{$G'_{s_1}$ and a maximizing alternating path (dashed lines)}
\label{subfig:minus_s}
\end{subfigure}
\caption{\textbf{\subref{subfig:G})} $G$ with the smaller vertex set $U$ of size $s=2$ on the left and the larger set $V$ on the right. \textbf{\subref{subfig:G'})} We compute a \MWMk{s} on $G'$ to obtain a \MWM on $G$. \textbf{\subref{subfig:minus_t})} The matched vertex $t_2\in U$ is removed. We find a new \MWMk{s} from an augmenting path of maximal weight starting from that vertex' matching partner $s_1$. \textbf{\subref{subfig:minus_s})} The matched vertex $s_1\in V$ is removed. We find a new \MWMk{s-1} from an even length alternating path of maximal weight starting from that vertex' partner $t_2$.}
\label{fig:matchings}
\end{figure}

\begin{lemma}
\label{lemma:similargraph_matchings}
\label{lemma:speedup_Matching}
Let $G$ be a weighted bipartite graph with vertex sets $U$ and $V$, $s:=|U|\leq|V|=:t$.
Let either $C=U$ or $C=V$.
We can compute a \MWM on $G$ and a \MWM on each graph $G_c$, $c\in C$ 
in time $\mathcal O(s^4+s^2t)$.
\end{lemma}

\begin{proof}
From Lemma~\ref{lemma:matching_real_single} we know there is an algorithm which computes a \MWM on $G$ in time $\mathcal O(s^2t)$.
This algorithm first copies the $s$ vertices of $U$ and then adds an edge of weight 0 between each vertex of $U$ and its copy.
We denote this graph $G'$.
The algorithm computes a \MWMk{s} $M'$ on $G'$ ($M'$ is also a \MWM), which corresponds to a \MWM on $G$.
The graph $G'$ with one vertex $c\in N(v)$ removed is denoted $G'_c$.
If $c$ is not matched, we are done.
If $c$ is matched, let $u$ be the partner of $c$.
Let $M'_c:=M'\setminus \{cu\}$.
We observe $|M'_c|=s-1$.

First, assume $c\in V$.
In a \MWMk{s} of $G'_c$ each vertex of $U$ including $u$ must be matched.
An uneven length $M'_c$-augmenting path $P$ of maximal weight (the path's weight refers to the difference in the matching's weight after augmenting) incident to $u$ yields a \MWMk{s} on $G'_c$ and thus a \MWM on $G_c$.
This follows from the fact, that any $M'_c$-alternating cycle or path on $G'_c$ not incident to $u$ has nonpositive weight; otherwise $M'$ was no \MWM.
We can find such a path $P$ using Bellman-Ford in time $\mathcal O(st+s^3)$.
To this end, we multiply the weights of non matching edges by -1.
I.e., the absolute value of the length of a shortest path (which is nonpositive) corresponds to the increase in the matching's weight.
Since at most $s$ vertices of $V$ are matched by $M$, the total time is $\mathcal O(s^2t+s^4)$.

Second, assume $c\in U$.
Since $c$ is matched, let $v$ be the partner of $c$. 
$M'_c$ is of cardinality $s-1$.
This time we need to find an alternating path of even length (we removed a vertex from $U$) and of maximal weight incident to $v$.
Any alternating cycle or path not incident to $v$ cannot augment $M'_c$ to greater weight; otherwise $M'$ was no \MWM.
This path may have length 0, e.g., if $M'_c$ is a MWM of $M'$.
Again, the total time to compute the paths is $\mathcal O(s^2t+s^4)$.
Figure~\ref{fig:mcs_edge} illustrates the computation of the \MWM{}s.
\end{proof}

\subparagraph{Sequence of computation.}
For given vertices $(u,v)\in V(T)\times V(T')$ we denote the matching graph $G$ without removed vertices as \emph{main instance} and the matching graphs $G_c$ as its \emph{sub instances}.
Analog for \tin we define $M^{T',v}_{t}$ as main instance and $M^{T',c}_{t}$ for each $c\in N(v)$ as its \emph{sub instances}.

Let $u\in V(T)$ and $vw\in E(T')$.
We observe, for type \tin the following values depend circular on each other.
$M^{T',w}_{t}$ for the pair $(u,w)$ requires $\TabL^w(u,v,t)$ requires $M^{T',v}_{t}$ for the pair $(u,v)$ requires $\TabL^v(u,w,t)$ requires $M^{T',w}_{t}$ for the pair $(u,w)$.
We further observe, the \MWM{s} depend on table entries of both types.
We may break the dependencies by solving at most one sub instance before solving the main instance, as shown next.

We iterate over all roots $s\in V(T')$ and compute a rooted \WES between $T^r$ and $T'^s$ as in Lemma~\ref{lemma:recursion}.
During the recursion on vertices $(u,v)$ the following cases may happen.
\begin{enumerate}
\item The first instance to compute on $(u,v)$ is a main instance. Then we instantly compute all its sub instances from it.
\item The first instance to compute on $(u,v)$ is a sub instance. Then we compute only the sub instance without deriving it from the main instance.
\begin{enumerate}
\item If the second instance is a main instance, we instantly compute its sub instances.
\item Otherwise let $c_1\in N(v)$ and $c_2\in N(v)$ be the vertices corresponding to the first and second sub instance, respectively.
Let us consider table entries first.
When we computed the sub instance corresponding to $c_1$, all necessary table entries for $M^{T',v}_{t}$ except $\TabL^v(u,c_1,t)$ were available.
For the second sub instance $\TabL^v(u,c_1,t)$ is also available.
Thus we may instantly compute the main instance and all other sub instances including the one corresponding to $c_2$.
We may argue analog for the \MWM{s}.
\end{enumerate}
\end{enumerate}

\begin{theorem}
Let $T$ and $T'$ be (unrooted) vertex and/or edge labeled trees.
Let $\wIso$ be a weight function, $\Delta=\min \{\Delta(T),\Delta(T')\}$, and $p$ be a distance penalty.
A \WESu between $T$ and $T'$ can be computed in time $\mathcal{O}(|T|\,|T'|\Delta\}$ and space $\mathcal O(|T|\,|T'|)$.
\end{theorem}

\section{Conclusions}
We presented an algorithm which solves the largest weighed common subtree embedding problem in time $\mathcal O(kl\Delta)$.
For rooted trees of integral weights bounded by a constant we proved a bound of $\mathcal O(|T|\,|T'|\sqrt\Delta\log (C\min\{|T|,|T'|\}))$.
Our approach generalizes the maximum common subtree problem~\cite{DrKrMu2016} and the largest common subtree embedding problem both unlabeled~\cite{GuNi1998} and labeled~\cite{Kao2001} by supporting weights between labels and a distance penalty for skipped vertices.

A remaining open problem is whether the time bound for unrooted trees can be improved when the weights are integral and bounded by a constant.
Since weight scaling algorithms for matchings do not work incrementally~\cite{RaTa12}, there is no obvious way to exploit the similarities in the given matching graphs.

\bibliography{lit}
\end{document}